\documentclass[submission,copyright,creativecommons,fleqn]{eptcs}
\providecommand{\event}{GandALF 2012} 
\usepackage{breakurl}             

\usepackage{amsmath}
\usepackage{amssymb}
\usepackage{amsthm}
\newtheorem{lemma}{Lemma}
\newtheorem{definition}{Definition}
\newtheorem{theorem}{Theorem}
\newtheorem{numberedremark}{Remark}
\newtheorem{corollary}{Corollary}
\newtheorem{example}{Example}

\usepackage{tikz}
\usetikzlibrary{decorations.pathreplacing}
\usetikzlibrary{decorations.pathmorphing}
\usetikzlibrary{shapes.geometric}
\usetikzlibrary{arrows}

\tikzset{p0/.style = {shape = circle,    draw, thick, minimum size = 0.6cm}}
\tikzset{p1/.style = {shape = rectangle, draw, thick, minimum size = 0.6cm}}

\tikzset{>=stealth}

\tikzset{every edge/.style = {thick, ->, draw}}
\tikzset{every loop/.style = {thick, ->, draw}}

\renewcommand{\epsilon}{\varepsilon}
\newcommand{\F}{\ensuremath{\mathcal{F}}}
\newcommand{\pow}[1]{2^{#1}}
\newcommand{\nats}{\mathbb{N}}
\newcommand{\mem}{\mathfrak{M}}
\newcommand{\game}{\mathcal{G}}
\newcommand{\win}{\mathrm{Win}}
\newcommand{\arena}{\mathcal{A}}

\newcommand{\prefs}{\mathrm{Pref}}
\newcommand{\behavior}{\mathrm{Beh}}
\newcommand{\aut}{\mathfrak{A}}

\newcommand{\smallerfone}{\le_{\F_1}}
\newcommand{\smallerf}{\le_{\F}}
\newcommand{\equalfone}{=_{\F_1}}
\newcommand{\equalf}{=_{\F}}
\newcommand{\equalclass}[1]{[#1]}
\newcommand{\lar}{\mathrm{LAR}}

\DeclareMathOperator{\infi}{\mathrm{Inf}}
\DeclareMathOperator{\occ}{\mathrm{Occ}}
\DeclareMathOperator{\maxscore}{\mathrm{MaxSc}}

\DeclareMathOperator{\score}{\mathrm{Sc}}
\DeclareMathOperator{\acc}{\mathrm{Acc}}
\DeclareMathOperator{\plays}{\mathrm{Plays}}
\DeclareMathOperator{\last}{\mathrm{Last}}
\DeclareMathOperator{\update}{\mathrm{Upd}}
\DeclareMathOperator{\nxt}{\mathrm{Nxt}}
\DeclareMathOperator{\init}{\mathrm{Init}}

\title{Down the Borel Hierarchy:\\ Solving Muller Games via Safety Games\thanks{This work was supported by the projects \textit{Games for Analysis and
Synthesis of Interactive
Computational Systems (GASICS)} and \textit{Logic for Interaction (LINT)} of the
\textit{European Science
Foundation}.}}

\author{
Daniel Neider
\institute{Lehrstuhl f{\"u}r Informatik 7\\RWTH Aachen University, Germany}
\email{neider@automata.rwth-aachen.de}
\and
Roman Rabinovich
\institute{Mathematische Grundlagen der Informatik\\RWTH Aachen University, Germany}
\email{rabinovich@logic.rwth-aachen.de}
\and
Martin Zimmermann
\institute{
RWTH Aachen University, Germany \&
University of Warsaw, Poland}
\email{zimmermann@mimuw.edu.pl}
}

\def\titlerunning{Solving Muller Games via Safety Games}
\def\authorrunning{D. Neider, R. Rabinovich \& M. Zimmermann}

\begin{document}
\maketitle

\begin{abstract}
We transform a Muller game with $n$ vertices into a safety game with $(n!)^3$
vertices whose solution allows to determine the winning regions of the Muller
game and to compute a finite-state winning strategy for one player. This yields
a novel antichain-based memory structure and a natural notion of permissive
strategies for Muller games. Moreover, we generalize our construction by
presenting a new type of game reduction from infinite games to safety games and
show its applicability to several other winning conditions.
\end{abstract}

\section{Introduction}
\label{sec_intro}

Muller games are a source of interesting and challenging questions in the theory
of infinite games. They are expressive enough to describe all $\omega$-regular
properties. Also, all winning conditions that depend only on the set of vertices
visited infinitely often can trivially be reduced to Muller games. Hence, they
subsume B\"uchi, co-B\"uchi, parity, Rabin, and Streett conditions. Furthermore,
Muller games are not positionally determined, i.e., both players need memory to
implement their winning strategies. In this work, we consider three aspects of
Muller games: solution algorithms, memory structures, and quality measures for
strategies.

To date, there are two main approaches to solve Muller games: direct algorithms
and reductions. Examples for the first approach are Zielonka's recursive
polynomial space al\-go\-rithm~\cite{Z98}, which is based on earlier work by
McNaughton~\cite{M93}, and Horn's polynomial time algorithm for explicit Muller
games \cite{H08}. The second approach is to reduce a Muller game to a parity
game using Zielonka trees~\cite{DJW97} or latest appearance records
(LAR)~\cite{GH82}.

In general, the number of memory states needed to win a Muller game is
prohibitively large~\cite{DJW97}. Hence, a natural task is to reduce this number
(if possible) and to find new memory structures which may implement small
winning strategies in subclasses of Muller games.

As for the third aspect, to the best of our knowledge there is no previous work
on quality measures for strategies in Muller games. This is in contrast to other
winning conditions. Recently,  much attention is being paid to not just
synthesize some winning strategy, but to find an optimal one according to a
certain quality measure, e.g., waiting times in request-response
games~\cite{HTW08} and their extensions~\cite{Z09}, permissiveness in parity
games~\cite{BJW02, BMOU11}, and the use of weighted automata in quantitative
synthesis~\cite{BCHJ09, CCHRS11}.

Inspired by work of McNaughton~\cite{M00}, we present a framework to deal with
all three issues. Our main contributions are a novel algorithm and a novel type
of memory structure for Muller games. We also obtain a natural quality measure
for strategies in Muller games and are able to extend the definition of
permissiveness to Muller games.

While investigating the interest of Muller games for ``casual living-room
recreation''~\cite{M00}, McNaughton introduced scoring functions which describe
the progress a player is making towards winning a play: consider a Muller game
$(\arena, \F_0, \F_1)$, where $\arena$ is the arena and $(\F_0, \F_1)$ is a
partition of the set of loops in $\arena$ used to determine the winner. Then, the score
of a set~$F$ of vertices measures how often $F$ has been visited completely
since the last visit of a vertex not in~$F$. Player~$i$ wins a play if and only
if there is an $F\in \F_i$ such that the score of~$F$ tends to infinity while
being reset only finitely often (a reset occurs whenever a vertex outside $F$ is
visited).

McNaughton proved the existence of strategies for the winning player that bound
her opponent's scores by $|\arena|!$~\cite{M00}, provided the play starts in her
winning region. The characterization above implies that such a strategy is
necessarily winning. The bound $|\arena|!$ was subsequently improved to $2$ 
(and shown to be tight)~\cite{FZ12}. Since some score eventually
reaches value $3$, the winning regions of a Muller game can be determined by
solving the reachability game in which a player wins if she is the first to
reach a score of $3$~\footnote{ This reachability game was the object of study
in McNaughton's investigation of humanly playable games, which, for practical
reasons, should end after a bounded number of steps.}. However, it is cumbersome
to obtain a winning strategy for the infinite-duration Muller game from a
winning strategy for the finite-duration reachability game. The reason is that
one has to carefully concatenate finite plays of the reachability game to an
infinite play of the Muller game: reaching a score of~$3$ infinitely often does
not prevent the opponent from visiting other vertices infinitely often.

The ability to bound the losing player's scores can be seen as a safety
condition as well. This allows us to devise an algorithm to solve Muller games
that computes both winning regions and a winning strategy for one player.
However, we do not obtain a winning strategy for the other player. In general,
it is impossible to reduce a Muller game to a safety game whose solution yields
winning strategies for both players, since safety conditions are
on a lower level of the Borel hierarchy than Muller conditions.

Given a Muller game, we construct a safety game in which the scores of
Player~$1$ are tracked (up to score $3$). Player~$0$ wins the safety game, if
she can prevent Player~$1$ from ever reaching a score of~$3$. This allows to
compute the winning region of the Muller game by solving a safety game.
Furthermore, by exploiting the intrinsic structure of the safety game's arena we
present an antichain-based memory structure for Muller games. Unlike the memory
structures induced by Zielonka trees, which disregard the structure of the
arena, and the ones induced by LARs, which disregard the structure of the
winning condition $(\F_0, \F_1)$, our memory structure takes both directly into account:
a simple arena or a simple winning condition should directly lead to a small memory. The other two
structures only take one source of simplicity into account, the other one can only
be exploited when the game is solved.
Furthermore, our memory also implements the most general non-deterministic
winning strategy among those that prevent the opponent from reaching a certain
score in a Muller game. Thus, our framework allows to extend the notion of
permissiveness from positionally determined games to games that require memory.

Our idea of turning a Muller game into a safety game can be generalized to other
types of winning conditions as well. We define a weaker notion of reduction from
infinite games to safety games which not only subsumes our construction but
generalizes several constructions found in the literature. Based on work on
small progress measures for parity games~\cite{J00}, Bernet, Janin, and
Walukiewicz showed how to determine the winning regions in a parity game and a
winning strategy for one player by reducing it to a safety game~\cite{BJW02}.
Furthermore, Schewe and Finkbeiner~\cite{SF07} as well as Filiot, Jin, and Raskin~\cite{FJR11}
used a translation from co-B\"uchi games to safety
games in their work on bounded synthesis and LTL realizability, respectively. We
present further examples and show that our reduction allows to determine the
winning region and a winning strategy for one player by solving a safety game.
Thus, all these games can be solved by a new type of reduction and an algorithm
for safety games. Our approach simplifies the winning condition of the game --
even down the Borel hierarchy. However, this is offset by an increase in the
size of the arena. Nevertheless, in the case of Muller games, our arena is only
cubically larger than the arena constructed in the reduction to parity games.
Furthermore, a safety game can be solved in linear time, while the question
whether there is a polynomial time algorithm for parity games is open.

\section{Definitions}
\label{sec_defs}

The power set of a set $S$ is denoted by $\pow{S}$ and $\nats$ denotes the set
of non-negative integers. The prefix relation on words is denoted by
$\sqsubseteq$. For $\rho \in V^\omega$ and $L \subseteq V^\omega$ we define
$\prefs(\rho) = \{ w\in V^* \mid w \sqsubseteq \rho\}$ and $\prefs(L) =
\bigcup_{\rho \in L}\prefs(\rho)$. For $w=w_1\cdots w_n$,
let $\last(w)=w_n$.

An arena $\arena = (V, V_0, V_1, E)$ consists of a finite, directed graph
$(V,E)$ without terminal vertices, $V_0 \subseteq V$ and $V_1 = V \setminus
V_0$, where $V_i$ denotes the positions of Player~$i$. We require every vertex
to have an outgoing edge to avoid the nuisance of dealing with finite plays. The
size $|\arena|$ of $\arena$ is the cardinality of $V$. A loop $C \subseteq V$ in
$\arena$ is a strongly connected subset of $V$, i.e., for every $v,v'\in C$
there is a path from $v$ to $v'$ that only visits vertices in $C$.
A play in $\arena$ starting in $v\in V$ is an infinite sequence $\rho = \rho_0
\rho_1 \rho_2 \ldots$ such that $\rho_0 = v$ and $(\rho_n, \rho_{n+1}) \in E$
for all $n \in \nats$. The occurrence set $\occ( \rho )$ and infinity set
$\infi( \rho )$ of $\rho$ are given by $\occ( \rho ) = \{ v\in V \mid \exists n
\in \nats \text{ such that } \rho_n=v\}$ and $\infi( \rho ) = \{ v\in V \mid
\exists^{\omega} n \in \nats \text{ such that } \rho_n = v\}$. We also use the
occurrence set of a finite play infix $w$, which is defined in the same way. The
infinity set of a play is always a loop in the arena. A game $\game = ( \arena,
\win )$ consists of an arena $\arena$ and a set $\win \subseteq V^\omega$ of
winning plays for Player~$0$. The set of winning plays for Player~$1$ is
$V^\omega \setminus \win$.

A strategy for Player~$i$ is a mapping $\sigma \colon V^*V_i \rightarrow V$ such
that $(v, \sigma(wv)) \in E$ for all $wv \in V^* V_i$. We say that $\sigma$ is
positional if $\sigma(wv) = \sigma(v)$ for every $wv \in V^*V_i$.  A play
$\rho_0 \rho_1 \rho_2 \ldots$ is consistent with $\sigma$ if $\rho_{n+1} =
\sigma( \rho_0 \cdots \rho_n)$ for every~$n$ with $\rho_n \in V_i$. For $v \in
V$ and a strategy $\sigma$ for one of the players, we define the behavior of
$\sigma$ from $v$ by $\behavior( v, \sigma ) = \{\rho \in V^\omega \mid \rho
\text{ play starting in $v$ that is consistent with $\sigma$}\}$. A
strategy~$\sigma$ for Player~$i$ is a winning strategy from a set of vertices $W
\subseteq V$ if every $\rho \in \behavior(v, \sigma)$ for $v \in W$ is won by Player~$i$. The
winning region~$W_i( \game )$ of Player~$i$ in $\game$ contains all vertices
from which Player~$i$ has a winning strategy. We always have $W_0( \game ) \cap
W_1( \game ) = \emptyset$ and $\game$ is determined if~$W_0(\game) \cup
W_1(\game) = V$. A winning strategy for Player~$i$ is uniform, if it is winning
from all $v \in W_i(\game)$.

A memory structure~$\mem = (M, \init, \update)$ for an arena $(V, V_0, V_1, E )$
consists of a finite set~$M$ of memory states, an initialization function~$\init
\colon V \rightarrow M$, and an update function~$\update\colon M\times
V\rightarrow M$. The update function can be extended to $\update^*\colon
V^+\rightarrow M$ in the usual way: $\update^*( \rho_0 ) = \init( \rho_0 )$ and
$\update^* ( \rho_0 \ldots \rho_n \rho_{n+1} ) = \update( \update^*( \rho_0
\ldots \rho_n), \rho_{n+1})$. A next-move function (for Player~$i$) 
$\nxt \colon V_i \times M \rightarrow V$ has to satisfy $(v, \nxt(v,
m)) \in E$ for all $v \in V_i$ and all $m \in M$. It induces a strategy~$\sigma$
for Player~$i$ with memory~$\mem$ via $\sigma(\rho_0\ldots\rho_n)=\nxt( \rho_n,
\update^*( \rho_0 \ldots \rho_n))$. A strategy is called finite-state if it can
be implemented with a memory structure. The size of $\mem$ (and, slightly
abusive, $\sigma$) is $|M|$. 

We consider two types of games defined by specifying $\win$ implicitly. A safety
game is a tuple $\game = ( \arena, F )$ with $F \subseteq V$ and $\win = \{ \rho
\in V^\omega \mid \occ(\rho) \subseteq F \}$. A Muller game is a tuple $\game =
( \arena, \F_0, \F_1 )$ where $\F_0$ is a set  of loops of $\arena$, $\F_1$
contains the loops which are not in $\F_0$, and $\win = \{ \rho \in V^\omega
\mid \infi(\rho) \in \F_0 \}$, i.e., $\rho$ is winning for Player~$i$ if and
only if $\infi( \rho ) \in \F_i$. Safety games are determined with uniform
positional strategies and Muller games are 
determined with uniform finite-state strategies of size $|\arena|!$~\cite{GH82}.
\section{Scoring Functions for Muller Games}
\label{sec_score}

In this section, we introduce scores and accumulators for Muller
games. These concepts describe the progress of a player throughout a play.
Intuitively, for each set $F \subseteq V$, the score of $F$ of a play prefix $w$
measures how often $F$ has been visited completely since the last visit of a
vertex that is not in $F$ or since the beginning of $w$. The accumulator of the
set $F$ measures the progress made towards the next score increase: $\acc_F(w)$
contains the vertices of~$F$ seen since the last increase of the score of $F$ or
the last visit of a vertex $v \notin F$, depending on which occurred later. For
a more detailed treatment we refer to~\cite{FZ12,M00}.

\begin{definition}
Let $w\in V^+$, $v\in V$, and $F\subseteq V$.
\begin{itemize}
\item Define $\score_{\{v\}}(v)=1$ and $\acc_{\{v\}}(v)=\emptyset$, and for
 $F\not= \{v\}$ define $\score_F(v)=0$ and $\acc_F(v)=F\cap\{v\}$.
\item If $v\notin F$, then $\score_F(wv) = 0$ and $\acc_F(wv) = \emptyset$.
\item If $v\in F$ and $\acc_F(w) = F\setminus \{v\}$, then $\score_F(wv) =
\score_F(w) +1 $ and
$\acc_F(wv) = \emptyset$.
\item If $v\in F$ and $\acc_F(w) \not= F\setminus \{v\}$, then $\score_F(wv) =
\score_F(w)$ and
$\acc_F(wv) = \acc_F(w) \cup \{v\}$.
\end{itemize}
Also, for $\F \subseteq \pow{V}$ define $\maxscore_{\F}\colon V^+\cup V^{\omega}\rightarrow
\nats\cup\{\infty\}$ by $\maxscore_{\F}(\rho) = \max_{F \in \F} \max_{w
\sqsubseteq \rho} \score_F(w)$.
\end{definition}

\begin{example}
Let $V = \{0,1,2\}$ and $F = \{0,1\}$. We have $\score_F(10012100) = 1$ and
$\acc_F(10012100) = \{0\}$, but $\maxscore_{\{F\}}(10012100) = 2$, due to the
prefix $1001$. The score for $F$ is reset to $0$ by the occurrence of $2$, i.e.,
$\score_F(10012) = 0$ and $\acc_F(10012) = \emptyset$.
\end{example}

If $w$ is a play prefix with $\score_F(w)\ge 2$, then $F$ is a loop of the
arena. In an infinite play~$\rho$, $\infi(\rho)$ is the unique set $F$ such that
$\score_F$ tends to infinity while being reset to $0$ only finitely often.
Hence, $\maxscore_{ \F_{1-i}}(\rho) < \infty$ implies $\infi(\rho) \in \F_i$.
Also, we always have $\acc_F(w) \subsetneq F$.

Next, we give a score-based preorder and an induced equivalence relation on play
prefixes.

\begin{definition}
\label{def_smaller}
Let $\F \subseteq \pow{V}$ and $w,w' \in V^+$.
\begin{enumerate}
\item $w$ is $\F$-smaller than $w'$, denoted by $w \smallerf w'$, if $\last(w) =
\last(w')$ and for all $F \in \F$:
\begin{itemize}
\item $\score_F(w) < \score_F(w')$, or
\item $\score_F(w) = \score_F(w')$ and $\acc_F(w) \subseteq \acc_F(w')$.
\end{itemize}
\item $w$ and $w'$ are $\F$-equivalent, denoted by $w \equalf w'$, if
$w\smallerf w'$ and $w'\smallerf w$.
\end{enumerate}
\end{definition}

The condition $w \equalf w'$ is equivalent to $\last(w) = \last(w')$ and for
every $F \in \F$ the equalities $\score_F(w) = \score_F(w')$ and $\acc_F(w) =
\acc_F(w')$ hold. Thus, $\equalf$ is an equivalence relation. Both $\smallerf$
and $\equalf$ are preserved under concatenation, i.e., $\equalf$ is a
congruence.

\begin{lemma}
\label{lem_concat} Let $\F \subseteq \pow{V}$ and $w,w' \in V^+$.
\begin{enumerate}
\item\label{lem_smaller_concat} If $w \smallerf w'$, then $wu \smallerf w'u$ for
all $u \in V^*$.
\item\label{lem_equal_concat} If $w \equalf w'$, then $wu \equalf w'u$ for all
$u \in V^*$.
\end{enumerate}
\end{lemma}
\section{Solving Muller Games by Solving Safety Games}
\label{sec_muller2safety}

In this section, we show how to solve a Muller game by solving a safety game.
Our approach is based on the existence of winning strategies for Muller games
that bound the losing player's scores by $2$.

\begin{lemma}[\cite{FZ12}]
\label{lem_lt3}
In every Muller game $\game=(\arena, \F_0, \F_1)$ Player~$i$ has a winning
strategy~$\sigma$ from $W_i(\game)$ such that $\maxscore_{\F_{1-i}}(\rho)\le 2$
for every play $\rho \in \behavior(v, \sigma)$ with $v \in W_i(\game)$.
\end{lemma}
\newpage
The following example shows that the bound $2$ is tight. 
\begin{example}
\label{example_runningmuller}
Consider the Muller game~$\game = (\arena, \F_0, \F_1)$, where $\arena$ is
depicted in Figure~\ref{figure_runningmuller_arena}, $\F_0 =
\{\{0\},\{2\},\{0,1,2\}\}$ and $\F_1 = \{\{0,1\},\{1,2\}\}$. By alternatingly
moving from $1$ to $0$ and to $2$, Player~$0$ wins from every vertex, i.e., we
have $W_0(\game)= \{0,1,2\}$, and she bounds Player~$1$'s scores by $2$. However, he
is able to achieve a score of two: consider a play starting at $1$ and suppose
(w.l.o.g.) that Player~$0$ moves to vertex~$0$. Then, Player~$1$ uses the
self-loop once before moving back to $1$, thereby reaching a score of $2$ for
the loop $\{0,1\}\in \F_1$.

\begin{figure}[h]
\begin{center}
\begin{tikzpicture}[node distance = 1.5cm]
\node[p0]			(c)	{$1$};
\node[p1,left of = c]		(l)	{$0$};
\node[p1,right of = c]	(r)	{$2$};

\path
(c) edge[bend left] (l)
(l) edge[bend left] (c)
(c) edge[bend left] (r)
(r) edge[bend left] (c)
(l) edge[loop left] ()
(r) edge[loop right] ();
\end{tikzpicture}
\caption[The arena for Example~\ref{example_runningmuller}]{The arena $\arena$ for Example~\ref{example_runningmuller}}
\label{figure_runningmuller_arena}
\end{center}
\end{figure}
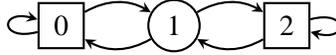
\end{example}

A simple consequence of Lemma~\ref{lem_lt3} is that a vertex $v$ is in
Player~$0$'s winning region of the Muller game $\game$ if and only if she can
prevent her opponent from ever reaching a score of $3$ for a set in $\F_{1}$.
This is a safety condition which only talks about small scores of one player. To
determine the winner of $\game$, we construct an arena which keeps track of the
scores of Player~$1$ up to threshold $3$. The winning condition~$F$ of the
safety game requires Player~$0$ to prevent a score of $3$ for her opponent.

\begin{theorem}
\label{thm_muller2safety}
Let $\game$ be a Muller game with vertex set $V$. One can effectively construct
a safety game $\game_S$ with vertex set $V^S$ and a mapping $f\colon
V\rightarrow V^S$ with the following properties:
\begin{enumerate}
\item\label{thm_muller2safety1} For every $v\in V$:  $v\in W_i(\game)$ if and
only if $f(v)\in W_i(\game_S)$.
\item\label{thm_muller2safety2} Player~$0$ has a finite-state winning strategy
from $W_0(\game)$ with memory $M \subseteq W_0(\game_S)$.
\item\label{thm_muller2safety3} $|V^S|\le \left(\sum_{k=1}^ {|V|}\binom{|V|}{k}\cdot
k!\cdot2^k\cdot k!\right)+1\le (|V|!)^3$.
\end{enumerate}
\end{theorem}

Note that the first statement speaks about both players while the second one
only speaks about Player~$0$. This is due to the fact that the safety game keeps
track of Player~$1$'s scores only, which allows Player~$0$ to prove that she can
prevent him from reaching a score of $3$. But as soon as a score of $3$ is
reached, the play is stopped. To obtain a winning strategy for Player~$1$, we
have to swap the roles of the players and construct a safety game which keeps
track of the scores of Player~$0$.
Alternatively, we could construct the arena which keeps track of both player's
scores. However, that would require to define two safety games in this arena:
one in which Player~$0$ has to avoid a score of $3$ for Player~$1$ and vice
versa. 
This arena is larger than the ones in which only the scores of one player
are tracked (but still smaller than $(|V|!)^3$).
It is well-known that it is
impossible to reduce a Muller game to a single safety game and thereby obtain
winning strategies for both players. We come back to this in
Section~\ref{sec_safetyreductions}.

We begin the proof of Theorem~\ref{thm_muller2safety} by defining the safety
game $\game_S$. Let $\game=(\arena,\F_0,\F_1)$ with arena $\arena = (V, V_0, V_1, E)$.
We define \[\plays_{< 3} = \{ w \mid \text{$w$ play prefix in $\game$ and
$\maxscore_{\F_1}(w) < 3$}\}\]
to be the set of play prefixes in $\game$ in which the scores of Player~$1$ are
at most $2$ and we define
\begin{align*}
\plays_{=3} = \{w_0\cdots w_{n+1} \mid &\text{$w_0\cdots w_{n+1}$ play
prefix in $\game$, $\maxscore_{\F_1}(w_0\cdots w_n )\le 2$, and} \\
&\text{$\maxscore_{\F_1}(w_0\cdots w_n w_{n+1})= 3$ }\}
\end{align*}
to be the set of play prefixes in which Player~$1$ just reached a score of $3$.
Furthermore, let $\plays_{\le 3} = \plays_{< 3}\cup \plays_{=3}$.
Note that these definitions ignore the scores of Player~$0$.
The arena of the safety game we are about to define is the $\equalfone$-quotient
of the unraveling of $\arena$ up to the positions where Player~$1$ reaches a
score of $3$ for the first time (if he does at all).
\newpage
Formally, we define $\game_S=((V^S,V_0^S,V_1^S,E^S),F)$ where

\begin{itemize}
\item $V^S=\plays_{\le 3}/_{\equalfone}$,
\item $V_i^S=\{\equalclass{w}_{\equalfone} \mid \equalclass{w}_{\equalfone}\in
V^S \text{ and } \last(w)\in V_i \}$ for $i \in \{0,1\}$,
\item $(\equalclass{w}_{\equalfone}, \equalclass{wv}_{\equalfone})\in E^S$ for
all $w\in\plays_{< 3}$ and all $v$ with $(\last(w), v)\in E$ \footnote{Hence,
every vertex in $\plays_{=3}$ is terminal, contrary to our requirements on an
arena. However, every play visiting these vertices is losing for Player~$0$ no
matter how it is continued. To simplify the following proofs, we refrain from
defining outgoing edges for these vertices.}, and
\item $F=\plays_{< 3}/_{\equalfone}$.
\end{itemize}

The definitions of $V_0^S$ and $V_1^S$ are independent of representatives, as
$w\equalfone w'$ implies $\last(w) = \last( w')$, we have $V^S = V_0^S \cup
V_1^S$ due to $V=V_0 \cup V_1$, and $F$ is well-defined, since every equivalence class in
$\plays_{< 3}/_{\equalfone}$ is also one in $\plays_{\le 3}/_{\equalfone}$. Finally, let
 $f(v) =
\equalclass{v}_{\equalfone}$ for every $v \in V$.

\begin{numberedremark}
\label{rem_edges}
If $(\equalclass{w}_{\equalfone}, \equalclass{w'}_{\equalfone})\in E^S$, then we have 
$(\last (w), \last(w')) \in E$.
\end{numberedremark}

For the sake of readability,
we denote $\equalfone$-equivalence classes
by $\equalclass{w}$ from now on. All definitions and statements below
are independent of representatives and we refrain from mentioning
it from now on.

\begin{example}
\label{ex_safety}
To illustrate these definitions, Figure~\ref{fig_examplesafety} depicts the safety game~$\game_S$ for the Muller game~$\game$ from Example~\ref{example_runningmuller}.
One can verify easily that the vertices~$\equalclass{v}$ for $v\in V$ are in the winning region of Player~$0$. This corresponds to the fact that Player~$0$'s winning region in the Muller game contains every vertex.

\begin{figure*}
\begin{center}
\begin{tikzpicture}[scale=.9]
\begin{scope}[node distance = 1cm]
\node (c) {};
\node[ellipse, thick, draw, double]		(1)	{$\equalclass{1}$};
\node[above of = c]				(placeu){};
\node[below of = c]				(placeb){};
\end{scope}

\begin{scope}[node distance = 1.5cm]
\node[p1, above of = placeu, double]		(0)	{$\equalclass{0$}};
\node[p1,below of = placeb, double]		(2)	{$\equalclass{2$}};
\end{scope}

\begin{scope}[node distance = 1.5cm]
\node[ellipse, thick, draw, right of = 0, double]	(01)	{$\equalclass{01}$};
\node[p1, right of = placeu, double]			(10)	{$\equalclass{10}$};
\node[p1, right of = placeb, double]			(12)	{$\equalclass{12}$};
\node[ellipse, thick, draw, right of = 2, double]	(21)	{$\equalclass{21}$};
\end{scope}

\begin{scope}[node distance = 2cm]
\node[ellipse, thick, draw, right of = 01, double]	(101)	{$\equalclass{101}$};
\node[p1, right of = 10, double]			(100)	{$\equalclass{100}$};
\node[p1, right of = 12, double]			(122)	{$\equalclass{122}$};
\node[ellipse, thick, draw, right of = 21, double]	(121)	{$\equalclass{121}$};
\end{scope}

\begin{scope}[node distance = 2.0cm]
\node[p1, right of = 101, double]			(1010)	{$\equalclass{1010}$};
\node[ellipse, thick, draw, right of = 100, double]	(1001)	{$\equalclass{1001}$};
\node[ellipse, thick, draw, right of = 122, double]	(1221)	{$\equalclass{1221}$};
\node[p1, right of = 121, double]			(1212)	{$\equalclass{1212}$};
\end{scope}

\begin{scope}[node distance = 2.5cm]
\node[ellipse, thick, draw, right of = 1010, double]	(10101)	{$\equalclass{10101}$};
\node[p1, right of = 1001, double]			(10010)	{$\equalclass{10010}$};
\node[p1, right of = 1221, double]			(12212)	{$\equalclass{12212}$};
\node[ellipse, thick, draw, right of = 1212, double]	(12121)	{$\equalclass{12121}$};
\end{scope}

\begin{scope}[node distance = 2.5cm]
\node[p1, right of = 10101]				(101010){$\equalclass{101010}$};
\node[ellipse, thick, draw, right of = 10010]		(100101){$\equalclass{100101}$};
\node[ellipse, thick, draw, right of = 12212]		(122121){$\equalclass{122121}$};
\node[p1, right of = 12121]				(121212){$\equalclass{121212}$};
\end{scope}

\path
(0) 	edge[loop below] 	()
(0) 	edge		(01)
(1) 	edge			(10)
(1) 	edge			(12)
(2) 	edge[loop above]	()
(2) 	edge			(21)
(01) 	edge[bend right = 40]	(12)
(10) 	edge			(100)
(10) 	edge			(101)
(12) 	edge			(122)
(12) 	edge			(121)
(21) 	edge[bend left = 40]		(10)
(101) 	edge[bend right = 12]	(12)
(101) 	edge			(1010)
(100) 	edge[loop below]	()
(100) 	edge			(1001)
(122) 	edge[loop above]	()
(122) 	edge			(1221)
(121) 	edge[bend left = 12]	(10)
(121) 	edge			(1212)
(1010) 	edge			(10101)
(1010) 	edge			(10010)
(1001) 	edge			(10010)
(1221) 	edge			(12212)
(1212) 	edge			(12212)
(1212) 	edge			(12121)
(10101)	edge			(101010)
(10010)	edge			(100101)
(10010)	edge[loop below]	()
(12212) edge			(122121)
(12212) edge[loop above]	()
(12121) edge			(121212);

\draw[-stealth, thick,  rounded corners](10101) |- ++(-9.6,.7) |-			(12);
\draw[-stealth, thick,  rounded corners](12121) |- ++(-9.9,-.7) |-			(10);

\draw(1001.south west) 	edge[bend left = 0]	(12);
\draw(1221.north west) 	edge[bend right = 0]	(10);

\draw(21.north) 	edge		(122);
\draw(01.south) 	edge		(100);

\draw[thick, rounded corners, dashed] (10.3,3.6) -- (10.3,1.8) -- (8.2,1.8) -- (7.2,2.4) -- (7.2,3.35) -- (5,3.35) -- (5,1.8) -- (7.3, 1.8) -- (7.3,1.4) --
(7.3,-1.4) -- (7.3, -1.8) -- (5,-1.8)-- (5,-3.35)-- (7.2,-3.35)-- (7.2,-2.4) -- (8.2,-1.8)-- (10.3,-1.8) -- (10.3,-3.6);

\node at (6.5,0) {$W_0(\game_S)$};
\node at (8.1,0) {$W_1(\game_S)$};

\end{tikzpicture}
\caption[The safety game for the Muller game from Example~\ref{example_runningmuller}]{The safety game~$\game_S$ for $\game$ from Example~\ref{example_runningmuller} (vertices drawn with double lines are in $F$); the dashed line separates the winning regions.}
\label{fig_examplesafety}
\end{center}
\end{figure*}
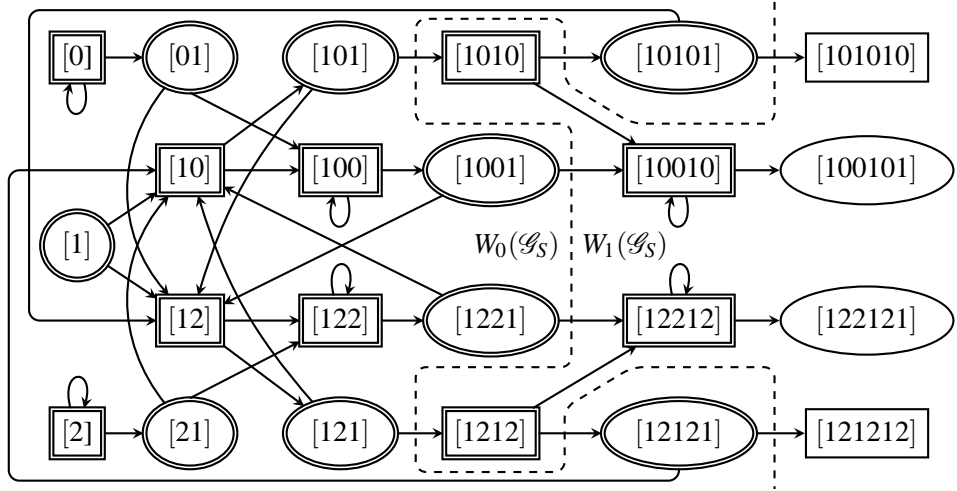
\end{example}

The proof
of Theorem~\ref{thm_muller2safety} is split into several lemmata. Due to
determinacy of both games, it suffices to consider only one Player (we pick $i=0$) to
prove Theorem~\ref{thm_muller2safety}.\ref{thm_muller2safety1}. 

To win the safety game, we simulate a winning strategy for the Muller game that bounds
Player~$1$'s scores by $2$, which suffices to avoid the vertices in $V^S \setminus F$,
which encode that a score of three is reached.

\begin{lemma}
\label{lem_equiv_l2r}
For every $v\in V$: if $v\in W_0(\game)$ then $\equalclass{v}\in W_0(\game_S)$.
\end{lemma}

For the other direction of
Theorem~\ref{thm_muller2safety}.\ref{thm_muller2safety1} we show that a subset
of $W_0(\game_S)$ can be turned into a memory structure for Player~$0$ in the
Muller game that induces a winning strategy. We use the $\equalfone$-equivalence
class of $w$ as memory state to keep track of Player~$1$'s scores in~$\game$.
But instead  of using all equivalence classes in the winning region of
Player~$0$, it suffices to consider the maximal ones with respect to
$\smallerfone$ that are reachable via a fixed positional winning strategy for
her in the safety game. Formally, we have to lift  $\smallerfone$ to equivalence
classes by defining $\equalclass{w}\smallerfone \equalclass{w'}$ if and only if
$w \smallerfone w'$.

The following proof is similar to the reductions from co-B\"{u}chi~\cite{
FJR11, KV05, SF07} and parity games~\cite{BJW02} to safety games, but for the more
general case of Muller games. We come back to the similarities to the latter
reduction when we want to determine permissive strategies in the next section.

\begin{lemma}
\label{lem_equiv_r2l}
For all $v\in V$: if $\equalclass{v}\in W_0(\game_S)$ then $v\in W_0(\game)$.
\end{lemma}

\begin{proof}
Let $\sigma$ be a uniform positional winning strategy for Player $0$ in
$\game_S$ and let $R \subseteq V^S$ be the set of vertices which are reachable
from $W_0(\game_S) \cap \{ \equalclass{v} \mid v \in V\}$ by plays consistent
with~$\sigma$. Every $\equalclass{w} \in R\cap V_0^S$ has exactly one successor
in $R$ (which is of the form $\equalclass{wv}$ for some $v \in V$) and dually,
every successor of $\equalclass{w} \in R\cap V_1^S$ (which are exactly the
classes $\equalclass{wv}$ with $(\last(w), v) \in E$) is in $R$.
Now, let $R_{\max}$ be the set of $\smallerfone$-maximal elements of $R$. Applying the
facts about successors of vertices in $R$ stated above, we obtain the following
remark.

\begin{numberedremark}
\label{rem_succ}
Let $R_{\max}$ be defined as above.
\begin{enumerate}
\item\label{rem_succ_p0} For every $\equalclass{w} \in R_{\max} \cap V_0^S$,
there is a $v \in V$ with $(\last(w), v) \in E$ and there is a $\equalclass{w' }
\in R_{\max}$ such that $\equalclass{wv} \smallerfone \equalclass{w'}$.
\item\label{rem_succ_p1} For every $\equalclass{w} \in R_{\max} \cap V_1^S$ and
each of its successors $\equalclass{wv}$, there is a $\equalclass{w'}\in
R_{\max}$ such that $\equalclass{wv} \smallerfone \equalclass{w'}$.
\end{enumerate}
\end{numberedremark}

Thus, instead of updating the memory from $\equalclass{w}$ to $\equalclass{wv}$
(and thereby keeping track of the exact scores) when processing a vertex~$v$, we
can directly update it to a maximal element that is $\F_1$-larger
than~$\equalclass{wv}$ (and thereby over-approximate the exact scores).
Formally, we define $\mem=(M, \init, \update)$ by $M=R_{\max} \cup \{\bot \}$
\footnote{We use the memory state~$\bot$ to simplify our proof. It is not
reachable via plays that are consistent with the implemented strategy and can
therefore be eliminated.},
\[ \init(v) = \begin{cases}
\equalclass{w}   &\text{if }\equalclass{v}\in W_0(\game_S) \text{ and there exists $\equalclass{w}\in R_{\max}$ with }\equalclass{v} \smallerfone
\equalclass{w},\\
\bot &\text{otherwise,}
\end{cases} \]
\[\update(\equalclass{w}, v) =
\begin{cases}
 \equalclass{w'}& \text{if there is some } \equalclass{w'} \in R_{\max} \text{ such that } \equalclass{wv} \smallerfone \equalclass{w'},\\
 \bot& \text{otherwise.}\end{cases}\]
This implies  $\equalclass{w} \smallerfone \update^*(w) $ for every $w \in V^+$
with $\update^*(w)\not= \bot$. Thus,  $\last(w) = \last(w')$, where
$\equalclass{w'} = \update^*(w)$.
Using Remark~\ref{rem_succ}, we define the next-move function 
\[\nxt(v, \equalclass{w}) = 
\begin{cases}
v' &\text{if }\last(w)=v\text{, }(v, v') \in E\text{, and there exists }
\equalclass{w'}\in R_{\max} \text{ such that }\equalclass{wv'} \smallerfone
\equalclass{w'}\text{,}\\[1ex]
v'' & \text{otherwise (where }v''\text{ is some vertex with }(v,v'')\in E\text{),}
\end{cases}\]
and $\nxt(v, \bot) = v''$ for some $v''$ with $(v,v'') \in E$. The second case
in the case distinction above is just to match the formal definition of a next-move
function; it is never invoked due to $\last(w) = \last(w')$ for
$\update^*(w) = \equalclass{w'}$ or $\update^*(w) = \bot$.

It remains to show that the strategy $\sigma$ implemented by $\mem$ and $\nxt$
is a winning strategy for Player~$0$ from $W = \{ v \mid \equalclass{v} \in
W_0(\game_S)\}$. An inductive application of Remark~\ref{rem_succ} shows that
every play $w$ that starts in $W$ and is consistent with $\sigma$ satisfies
$\update^*(w)\not = \bot$. This bounds the scores of Player~$1$ by $2$, as we
have $\equalclass{w} \smallerfone \update^*(w) \in R_{\max} \subseteq
\plays_{< 3}$ for every such play. Hence, $\sigma$ is indeed a winning
strategy for Player~$0$ from $W$.
\end{proof}

Using Lemma~\ref{lem_equiv_l2r} and the construction in the proof of
Lemma~\ref{lem_equiv_r2l} proves Theorem~\ref{thm_muller2safety}.\ref{thm_muller2safety2}.

\begin{corollary}
Player~$0$ has a finite-state winning strategy from $W_0(\game)$ whose memory
states form an $\smallerfone$-antichain in~$W_0(\game_S)$.
\end{corollary}

To finish the proof of Theorem~\ref{thm_muller2safety}, we determine the size of
$\game_S$ to prove the third statement. To this end, we use the concept of a
\emph{latest appearance record}~(LAR)~\cite{GH82,M93}. Note that we do not need
a hit position for our purposes.
A word $\ell\in V^+$ is an LAR if every vertex $v\in V$ appears at most once in
$\ell$. Next, we map each $w\in V^+$ to a unique LAR, denoted by $\lar(w)$, as
follows: $\lar(v)=v$ for every $v\in V$ and for $w\in V^+$ and $v\in V$ we
define $\lar(wv)= \lar(w)v$ if $v\notin\occ(w)$ and 
$\lar(wv)= p_1p_2v$ if $\lar(w)=p_1vp_2$. A simple induction shows that
$\lar(w)$ is indeed an LAR, which also ensures that the decomposition of $w$ in
the second case of the inductive definition is unique. We continue by showing
that $\lar(w)$ determines all but $|\lar(w)|$ many of $w$'s scores and
accumulators.

\begin{lemma}
\label{lem_lar}
Let $w\in V^+$ and $\lar(w)=v_k v_{k-1}\cdots v_1$.
\begin{enumerate}
\item\label{lem_lar1} $w = x_k v_k x_{k-1} v_{k-1} \cdots x_2 v_2 x_1 v_1$ for
some $x_i \in V^*$ with $\occ(x_i)\subseteq \{v_1,\ldots, v_i\}$ for every~$i$.
\item\label{lem_lar2} $\score_F(w)>0$ if and only if $F=\{v_1,\ldots, v_i\}$ for
some $i$.
\item\label{lem_lar3} If $\score_F(w)=0$, then $\acc_F(w)=\{v_1,\ldots,v_i\}$
for the maximal $i$ such that $\{v_1,\ldots,v_i\}\subseteq F$ and $\acc_F(w) =
\emptyset$ if no such $i$ exists.
\item\label{lem_lar4} Let $\score_F(w)>0$ and $F=\{v_1,\ldots, v_i\}$. Then,
$\acc_F(w)\in \{\emptyset\}\cup\{\{v_1,\ldots,v_j\}\mid j<i\}$.
\end{enumerate}
\end{lemma}

This characterization allows us to bound the size of $\game_S$ and to prove
Theorem~\ref{thm_muller2safety}.\ref{thm_muller2safety3}.

\begin{lemma}
We have $|V^S|\le \left(\sum_{k=1}^ {|V|}\binom{|V|}{k}\cdot
k!\cdot2^k\cdot k!\right)+1\le (|V|!)^3$.
\end{lemma}
\begin{proof}
In every safety game, we can merge the vertices in $V\setminus F$ to a single
vertex without changing $W_0(\game)$. Since $\equalclass{v}\in F$, we also
retain the equivalence $v\in W_i(\game)\Leftrightarrow \equalclass{v}\in
W_i(\game_S)$.

Hence, it remains to bound the index of $\plays_{<3}/_{\equalfone}$.
Lemma~\ref{lem_lar} shows that a play prefix~$w \in V^+$ has $|\lar(w)|$ many
sets with non-zero score. Furthermore, the accumulator of the sets with score
$0$ is determined by $\lar(w)$. Now, consider a play~$w \in \plays_{<3}$ and
a set~$F\in\F_1$ with non-zero score. We have $\score_F(w) \in \{1,2\}$ and
there are exactly $|F|$ possible values for $\acc_F(w)$ due to
Lemma~\ref{lem_lar}.\ref{lem_lar4}. Finally, $\lar(w) = \lar(w')$ implies
$\last(w) = \last(w')$.
Hence, the index of $\plays_{<3}/_{\equalfone}$ is bounded by the number of
$\lar$s, which is $\sum_{k=1}^n\binom{n}{k}\cdot k!$, times the number of
possible score and accumulator combinations for each $\lar$ $\ell$ of length~$k$, which is bounded by $2^k\cdot k!$.
\end{proof}

In the proof of Theorem~\ref{thm_muller2safety}.\ref{thm_muller2safety2}, we
used the maximal elements of Player~$0$'s winning region of the safety game that
are reachable via a fixed winning strategy. It is the choice of this strategy
that determines the size of our memory structure. However, finding a winning
strategy that visits at most $k \in \nats$ vertices in an arena (from a fixed
initial vertex) for a given~$k$ is NP-complete. This can be shown by a reduction
from the vertex cover problem (compare, e.g., \cite{Ehl11d} where a more general
result is shown). Moreover, it is not even clear that a small strategy also
yields few maximal elements.

In general, a player cannot prevent her opponent from reaching a score of
$2$, but there are arenas in which she can do so. By first constructing the subgame
$\game_S'$ up to threshold~$2$ (which is smaller than $\game_S$),
we can possibly determine a subset of Player~$0$'s winning region faster and
obtain a (potentially) smaller finite-state winning strategy for this subset.
But Example~\ref{example_runningmuller} shows that this approach is not complete.

\section{Permissive Strategies for Muller Games}
\label{sec_perm}

Bernet et al. introduced the concept of permissive strategies for parity
games\footnote{A parity game $(\arena, \Omega)$ consists of an arena~$\arena$
and a priority function $\Omega\colon V \rightarrow \nats$. A play $\rho$ is winning for
Player~$0$ if the minimal priority that is seen infinitely often during the play
is even.}~\cite{BJW02}, a non-deterministic winning strategy that subsumes the
behavior of every positional (non-deterministic) winning strategy. To compute
such a strategy, they reduce a parity game to a safety game. The main observation
underlying their reduction is the following: if we denote the number of vertices
of priority~$c$ by $n_c$, then a (non-deterministic) positional winning strategy
for Player~$0$ does not allow a play in which an odd priority~$c$ is visited
$n_c+1$ times without visiting a smaller priority in between. This property can
be formulated using scoring functions for parity games as well. The scoring
function for a priority~$c$ counts the occurrences of~$c$ since the last
occurrence of a smaller priority (such an occurrence resets the score for~$c$ to $0$).
Hence, our work on Muller games can be seen as a generalization of Bernet et
al.'s work. While the bound~$n_c$ on the scores in a parity game is
straightforward, the bound $2$ for Muller games is far from~obvious.

Since both constructions are very similar, it is natural to ask whether we can
use the concept of permissive strategies for Muller games. 
In parity games, we ask for a non-deterministic strategy that subsumes
the behavior of every \emph{positional} strategy. As positional
strategies do not suffice to win Muller games, we have to give 
a new definition of permissiveness for such games. In other words, we need to specify 
the strategies whose behaviors a permissive strategies for a Muller game should subsume.
One way to do this is to fix a
sufficiently large bound~$M$ and to require that a permissive strategy for a Muller
games subsumes the behavior of every finite-state winning strategy of size at
most $M$ (this was already proposed by Bernet et al. for parity
games~\cite{BJW02}).

However, we prefer to take a different approach. By closely inspecting the
reduction of parity to safety games, it becomes apparent that the induced
strategy does not only subsume the behavior of every positional winning
strategy, but rather the behavior of every strategy that prevents the opponent
from reaching a score of $n_c+1$ for some odd priority~$c$ (in terms of scoring
functions for parity games). It is this
formulation that we extend to Muller games: a (non-deterministic) winning
strategy for a Muller game is permissive, if it subsumes the behavior of every
(non-deterministic) winning strategy that prevents the losing player from
reaching a score of $3$. We formalize this notion in the following and show how
to compute such strategies from the safety game constructed in the previous
section.

A multi-strategy for Player~$i$ in an arena $(V, V_0, V_1, E)$ is a
mapping~$\sigma \colon V^*V_i \rightarrow \pow{V} \setminus \{\emptyset\}$ such
that $v' \in \sigma(wv)$ implies $(v,v') \in E$. A play~$\rho$ is consistent
with $\sigma$ if $\rho_{n+1} \in \sigma(\rho_0 \cdots \rho_n)$ for every $n$
such that $\rho_n \in V_i$. We still denote the plays starting in a
vertex~$v$ that are consistent with a multi-strategy~$\sigma$ by
$\behavior_\arena(v, \sigma)$ and define $\behavior_\arena(W, \sigma) =
\bigcup_{v \in W}\behavior_\arena(v, \sigma)$ for every subset~$W \subseteq V$. A
multi-strategy~$\sigma$ is winning for Player~$0$ from a set of vertices~$W$ in
a game~$(\arena, \win)$ if $\behavior_\arena(W, \sigma) \subseteq \win$, and a
multi-strategy~$\tau$ is winning for Player~$1$ from $W$, if
$\behavior_\arena(W, \sigma) \subseteq V^\omega \setminus \win$. It is clear
that the winning regions of a game do not change when we allow multi-strategies
instead of standard strategies.

To define finite-state multi-strategies we have to allow a next-move function to
return more than one vertex, i.e., we have $\nxt \colon V_i \times M \rightarrow
\pow{V} \setminus \{\emptyset\}$ such that $v' \in \nxt(v, m)$ implies $(v,v') \in E$.
A memory structure $\mem$ and $\nxt$ implement a multi-strategy~$\sigma$ via
$\sigma(wv) = \nxt(v, \update^*(wv))$.

\begin{definition}
A multi-strategy~$\sigma'$ for a Muller game~$\game$ is permissive, if 
\begin{enumerate}
\item $\sigma'$ is a winning strategy from every vertex in $W_0(\game)$, and
\item $\behavior_\arena(v, \sigma) \subseteq \behavior_\arena(v, \sigma')$ for every multi-strategy $\sigma$ and every vertex~$v$ with $\maxscore_{\F_1}(\rho) \le 2$ for every $\rho \in \behavior_\arena(v, \sigma)$.
\end{enumerate}
\end{definition}

The original definition for parity games replaces the second condition by the following requirement: $\behavior_\arena(v, \sigma) \subseteq \behavior_\arena(v, \sigma')$ for every positional multi-strategy $\sigma$ and every $v$ from which $\sigma$ is winning.

\begin{example}
Once again consider the Muller game of Example~\ref{example_runningmuller}. Starting at vertex~$1$, moving to $0$ is consistent with a winning strategy for Player~$0$ that bounds Player~$1$'s scores by $2$. Similarly, moving to $2$ is also consistent with a winning strategy for Player~$0$ that bounds Player~$1$'s scores. Hence, we have $\sigma'(1) = \{0,2\}$ for every permissive strategy~$\sigma'$. Now consider the play prefix $10$. Here it is Player~$1$'s turn and he can use the self-loop either infinitely often (which yields a play that is winning for Player~$0$) or only finitely often (say $n$ times) before moving back to vertex~$1$. In this situation, i.e., with play prefix $10^{n+1}1$, a strategy that bounds Player~$1$'s scores by $2$ has to move to vertex~$2$. Hence, we must have $\sigma'(10^{n+1}1) \supseteq \{2\}$. However, it is possible that we also have $1 \in \sigma'(10^{n+1}1)$, since a permissive strategy may allow more plays than the ones of strategies that bound Player~$1$'s scores by $2$. However, at some point, $\sigma'$ has to disallow the move back to vertex~$0$, otherwise it would allow a play that is losing for her. 
\end{example}

Using the safety game~$\game_S$ defined in the previous section, we are able to
show that Player~$0$ always has a finite-state permissive strategy and how to
compute one.

\begin{theorem}
Let $\game$ be a Muller game and $\game_S$ the corresponding safety game as
above. Then, Player~$0$ has a finite-state permissive strategy for $\game$ with memory
states $W_0(\game_S)$.
\end{theorem}

The proof is very similar to the one for
Theorem~\ref{thm_muller2safety}.\ref{thm_muller2safety2} (cf. the construction
in the proof of Lemma~\ref{lem_equiv_r2l}),
but we have to use all vertices in $W_0(\game_S)$ as memory states to implement
a permissive strategy, only using the maximal ones (restricted to those reachable
by some fixed winning strategy for the safety games) does not suffice.
Furthermore, the next-move function does not return one successor that guarantees
a memory update to a state from $W_0(\game_S)$, but it returns all such states.

\begin{proof}
We define $\mem = (M,
\init, \update)$ where $M = W_0(\game_S) \cup
\{\bot\}$~\footnote{Again, we use the memory state~$\bot$ to simplify our proof. It is
not reachable via plays that are consistent with the strategy implemented by
$\mem$ and can therefore eliminated and its incoming transitions can be redefined
arbitrarily.},
\begin{equation*}
\init(v) = \begin{cases}
\equalclass{v} &\text{if $\equalclass{v} \in W_0(\game_S)$,}\\
\bot					&\text{otherwise,}
\end{cases}
\end{equation*}
\begin{equation*}
\update(\equalclass{w}, v) = \begin{cases}
\equalclass{wv} &\text{if $\equalclass{wv} \in W_0(\game_S)$,}\\
\bot 			&\text{otherwise.}
\end{cases}
\end{equation*}
Hence, we have $\update^*(w) = \equalclass{w} \in W_0(\game_S)$ as long as every
prefix~$x$ of $w$  satisfies $\equalclass{x} \in W_0(\game_S)$, and
$\update^*(w) = \bot$ otherwise. We define $\nxt$ by $\nxt(v,
\bot) = \{v'\}$ for some successor~$v'$ of $v$ and
\begin{equation*}
\nxt(v, \equalclass{w})= \begin{cases}
\{ v' \mid \equalclass{wv'} \in W_0(\game_S)\} & \text{if $\equalclass{w} \in W_0(\game_S)$ and $\last(w) = v$,}\\
\{v'' \} &\text{otherwise, where $v''$ is some successor of $v$.}
\end{cases}
\end{equation*}
Since every vertex in $W_0(\game_S) \cap V_0^S$ has at least one successor in $W_0(\game_S)$, the next-move function always returns non-empty set of successors of $v$ in $\game$.

It remains to show that the strategy~$\sigma'$ implemented by $\mem$ and $\nxt$ is permissive. We begin by showing that $\sigma$ is winning from every vertex~$v \in W_(\game)$: due to Lemma~\ref{lem_equiv_l2r}, we have $\equalclass{v} \in W_0(\game_S)$. Hence, the memory is initialized with $\equalclass{v} \in W_0(\game_S)$. A simple induction shows $\update^*(w) = \equalclass{w} \in W_0(\game_S)$ for every play prefix that starts in $\equalclass{v}$ is consistent with $\sigma'$. This bounds Player~$1$'s scores by $2$. Hence, $\sigma'$ is indeed winning from $v$.

Finally, consider a multi-strategy $\sigma$ and a vertex~$v$ such that $\maxscore_{\F_1}(\rho) \le 2$ for every play~$\rho \in \behavior_\arena(v, \sigma)$. We have to show that every play~$\rho \in \behavior_\arena(v, \sigma)$ is consistent with $\sigma'$. Since $\sigma$ is winning from $v$ (as it bounds Player~$1$'s scores), we have $v \in W_0(\game)$. Now, assume $\rho$ is not consistent with $\sigma'$ and let $\rho_0 \cdots \rho_n \rho_{n+1}$ be the shortest prefix such that $\rho_{n+1} \notin \sigma'(\rho_0 \cdots \rho_n)$. Then, we have $\equalclass{\rho_0 \cdots \rho_n \rho_{n+1}} \notin W_0(\game_S)$. Hence, Player~$1$ has a strategy to enforce a visit to $V^S \setminus F$ in $\game_S$ starting in $\equalclass{\rho_0 \cdots \rho_n \rho_{n+1}}$. Player~$1$ can mimic this strategy in $\game$ to enforce a score of $3$ against every strategy of Player~$0$ when starting with the play prefix $\rho_0 \cdots \rho_n \rho_{n+1}$. Since this prefix is consistent with $\sigma$, which we have assumed to bound Player~$1$'s scores by $2$, we have derived the desired contradiction.
\end{proof}

\section{Safety Reductions for Infinite Games}
\label{sec_safetyreductions}

It is well-known that classical game reductions are not able to reduce Muller
games to safety games, since they induce continuous functions mapping (winning)
plays of the original game to (winning) plays of the reduced game. The existence
of such functions is tied to topological properties of the sets of winning plays
in both games. However, we transformed a Muller game to a safety game which
allowed us to determine the winning regions and a winning strategy for one
player. This is possible, since our reduction does not induce a continuous
function: a play is stopped as soon as Player~$1$ reaches a score of~$3$, but it
can (in general) be extended to be winning for Player~$0$. 

In this section, we briefly discuss the reason why Muller games can not be
reduced to safety games in the classical sense,  and then we present a novel
type of game reduction that allows us to reduce many games known from the
literature to safety games. The advantage of this safety reduction is that the
reduced game is always a safety game. Hence, we can determine the winning
regions of various games from different levels of the Borel hierarchy using the
same technique. However, we only obtain a winning strategy for one player, and
to give such a reduction, we need to have some information on the type of
winning strategies a player has in such a game. Let us begin by discussing
classical game reductions.

An arena $\arena = (V, V_0, V_1, E)$ and a memory structure $\mem = (M, \init,
\update)$ for $\arena$ induce the expanded arena $\arena\times\mem = (V \times
M, V_0 \times M, V_1 \times M, E' )$ where $((s,m), (v',m')) \in E'$ if and only
if $(v,v') \in E$ and $\update(m, v' ) = m'$. Every play $\rho$ in $\arena$ has
a unique extended play $\rho' = (\rho_0, m_0) (\rho_1, m_1) (\rho_2, m_2)
\ldots$ in $\arena \times \mem$ defined by $m_0 = \init( \rho_0 )$ and $m_{n+1}
= \update(m_n, \rho_{n+1})$, i.e., $m_n = \update^*(\rho_0 \cdots
\rho_n)$. A game $\game = ( \arena, \win)$ is reducible to $\game' = ( \arena',
\win')$ via $\mem$, written $\game \le_{ \mem } \game'$, if $\arena' = \arena
\times \mem$ and every play $\rho$ in $\game$ is won by the player who wins the
extended play $\rho'$ in $\game'$, i.e., $\rho \in \win$ if and only if
$\rho' \in \win'$.

\begin{lemma} \label{lem_reductiongivesmemory}
Let $\game$ be a game with vertex
set $V$ and $W \subseteq V$. If $\game \le_{ \mem } \game'$ and Player~$i$ has a
positional winning strategy for $\game'$ from $\{(v, \init(v))\mid v \in W\}$,
then she also has a finite-state winning strategy with memory~$\mem$ for $\game$
from $W$.
\end{lemma}

The set $\win_M \subseteq V^\omega$ of winning plays of a Muller game is in
general on a higher level of the Borel hierarchy than the set $\win_S \subseteq
U^\omega$ of winning plays of a safety game. Hence, in general, there exists no
continuous (in the Cantor topology) function $f\colon V^\omega \rightarrow
U^\omega$ such that $\rho \in \win_M$ if and only if $f(\rho) \in \win_S$~(see,
e.g.,~\cite{K95}). Since the mapping from a play in $\arena$ to its extended
play in $\arena\times\mem$ is continuous, we obtain the following negative
result (which holds for other pairs of games as well).

\begin{corollary}
\label{cor_nored}
In general, Muller games cannot be reduced to safety games.
\end{corollary}

To overcome this, we present a novel type of game reduction which encompasses
the construction presented in Section~\ref{sec_muller2safety} and is applicable
to many other infinite games.
\newpage
\begin{definition}
A  game $\game = (\arena, \win)$ with vertex set $V$ is (finite-state) safety reducible, if
there is a regular language $L \subseteq V^*$ of finite words such that:
\begin{itemize}
\item For every play $\rho \in V^\omega$: if $\prefs( \rho ) \subseteq L$, then
$\rho \in \win$.
\item If $v \in W_0(\game)$, then Player~$0$ has a strategy $\sigma$ such that
$\prefs( \behavior (v, \sigma )) \subseteq L$.
\end{itemize}
\end{definition}

Note that a strategy $\sigma$ satisfying $\prefs( \behavior (v, \sigma )) \subseteq L$ is
winning for Player~$0$ from $v$. Many games
appearing in the literature on infinite games are safety reducible, although these
reductions do neither yield fast solution algorithms nor optimal memory structures.

\begin{itemize}
\item In a B\"uchi game $\game$, Player~$0$ has a positional winning strategy
 such that every consistent play visits a vertex in $F$ at least every $k = |V\setminus F|$ steps. Hence,
$\game$ is safety reducible with $L = \prefs(((V\setminus F)^{\le k} \cdot
F)^\omega )$.

\item In a co-B\"uchi game $\game$, Player~$0$ has a positional winning strategy
 such that every consistent play stays in $F$ after visiting each vertex in $V\setminus F$ at most once.
Hence, $\game$ is safety reducible with $L = \prefs( \{w \cdot F^\omega \mid
\text{ each $v \in V\setminus F$ appears at most once in $w$}\})$.

\item In a request-response game $\game$, Player~$0$ has a finite-state winning
strategy  such that in every consistent play every request is answered within $k = |V|\cdot r\cdot 2^{r+1}$  steps,
where $r$ is the number of request-response pairs~\cite{HTW03}. Hence, $\game$
is safety reducible to the language of prefixes of plays in which every request
is answered within $k$ steps.
\item In a parity game, Player~$0$ has a positional winning
strategy  such that every consistent play does not visit $n_c + 1$ vertices with an odd priority $c$
without visiting a smaller even priority in between, where $n_c$ is the number
of vertices with priority c. Hence, $\game$ is safety reducible to the language
of prefixes of plays satisfying this condition for every odd priority $c$.
\item Lemma~\ref{lem_lt3} shows that a Muller game is safety reducible to the
language of prefixes of plays that never reach a score of $3$ for Player~$1$.
\end{itemize}

Let $\arena = (V, V_0, V_1, E)$ be an arena and let $\aut = (Q, V, q_0, \delta,
F)$ be a deterministic finite automaton recognizing a language over $V$. We
define the arena $\arena \times \aut = (V \times Q, V_0 \times Q, V_1 \times Q,
E\circ\delta)$ where $((v,q),(v',q')) \in E\circ \delta$ if and only $(v,v')
\in E$ and $\delta(q,v') = q'$. 

\begin{theorem}
\label{thm_safetyred}
Let $\game$ be a game with vertex set $V$ that is safety reducible with language
$L(\aut)$ for some DFA $\aut = (Q, V, q_0, \delta, F)$. Define the safety game
$\game' = (\arena \times \aut, V\times F)$.
\begin{enumerate}
\item\label{1} $v \in W_0(\game)$ if and only if $(v, \delta(q_0, v)) \in
W_0(\game')$.
\item\label{2} Player~$0$ has a finite-state winning strategy from $W_0(\game)$
with memory states $Q$.
\end{enumerate}
\end{theorem}

It is easy to show that every game in which Player~$0$ has a
finite-state winning strategy is safety-reducible to the prefixes of plays
consistent with this strategy, which is a regular language. However, for this
construction, we need a finite-state winning strategy, i.e., there is no need 
for a safety reduction. 

If $\game$ is determined, then Theorem~\ref{thm_safetyred}.\ref{1} is equivalent
to $v \in W_i(\game)$ if and only if $(v, \delta(q_0, v)) \in W_i(\game')$.
Hence, all games discussed above can be solved by
solving safety games. We conclude by mentioning that safety reducibility
of parity games was used implicitly to construct an algorithm for parity
games~\cite{J00} and to compute permissive strategies for parity
games~\cite{BJW02}. Furthermore, the safety reducibility of co-B\"uchi games is
used implicitly in work on bounded synthesis~\cite{SF07} and LTL
realizability~\cite{FJR11},  so-called ``Safraless'' constructions~\cite{KV05}
which the do not rely on determinization of automata on infinite words. 

Furthermore, the new notion of reduction allows to generalize permissiveness to all games
discussed in Example~\ref{ex_safety}: if a game is safety reducible to $L$, then
we can construct a multi-strategy that allows every play~$\rho$ in which
Player~$1$ cannot leave $L$ starting from any prefix of $\rho$. Thereby, we
obtain what one could call $L$-permissive strategies. For example, this allows to construct
the most general non-deterministic winning strategy in a request-response game that guarantees
a fixed bound on the waiting times.

\section{Conclusion}
\label{sec_conc}

We have shown how to translate a Muller game into a safety game to determine
both winning regions and a finite-state winning strategy for one player. Then,
we generalized this construction to a new type of reduction from infinite games
to safety games with the same properties. We exhibited several implicit
applications of this reduction in the literature as well as several new ones.
Our reduction from Muller games to safety games is implemented in the tool
\texttt{GAVS+}\footnote{See \url{http://www6.in.tum.de/~chengch/gavs/} for
details and to download the tool.}~\cite{CKLB11}. In the future, we want to
compare the performance of our solution to classical reductions to parity games
as well as to direct algorithms.

Our construction is based on the notion of scoring functions for Muller games.
Considering the maximal score the opponent can achieve against a strategy leads
to a hierarchy of all finite-state strategies for a given game. Previous
work has shown that the third level of this hierarchy is always non-empty, and
there are games in which the second level is empty. Currently, there is no
non-trivial characterization of the games whose first or second level of the
hierarchy is non-empty, respectively.

The quality of a strategy can be measured by its level in the hierarchy. We
conjecture that there is always a finite-state winning strategy of minimal size
in the least non-empty level of this hierarchy, i.e., there is no tradeoff
between size and quality of a strategy. This tradeoff may arise in many other
games for which a quality measure is defined. Also, a positive resolution of the
conjecture would decrease the search space for a smallest finite-state~strategy.

We used scores to construct a novel antichain-based memory structure for Muller
games. The antichain is induced by a winning strategy for the safety game. It is
open how the choice of such a strategy influences the size of the memory
structure and how heuristic approaches to computing winning strategies that only
visit a small part of the arena~\cite{daniel} influence the performance of our
reduction.

Finally, there is a tight connection between permissive strategies, progress
measure algorithms, and safety reductions for parity games: the progress measure algorithm due to
\mbox{Jurdzi\'{n}ski}~\cite{J00} and the reduction from parity games to safety
game due to Bernet et al.~\cite{BJW02} to compute permissive
strategies are essentially the same. Whether the safety reducibility of Muller
games can be turned into a progress measure algorithm is subject to ongoing
research.

\textbf{Acknowledgments} The authors want to thank Wolfgang Thomas for bringing McNaughton's
work to their attention, Wladimir Fridman for fruitful
discussions, and Chih-Hong Cheng for his implementation of the algorithm.

\bibliographystyle{eptcs}
\bibliography{ftMuller2}

\end{document}